\documentclass[a4paper,10pt]{article}
\usepackage{amsmath,amssymb,amsfonts,amsthm,latexsym,stmaryrd,eucal,a4} 
\usepackage{times} 
\usepackage{diagrams}
               
\renewcommand{\epsilon}{\varepsilon}
\renewcommand{\phi}{\varphi}

\newcommand{\N}{\mathbb{N}}

\newcommand{\cS}{\mathcal S}

\newcommand{\cH}{\mathcal H}

\newcommand{\cV}{\mathcal V}

\newcommand{\cF}{\mathcal F}

\newcommand{\cO}{\mathcal O}

\newcommand{\cM}{\mathcal M}

\newcommand{\cT}{\mathcal T}

\newcommand{\wh}{\widehat}

\newcommand{\display}{:}

\newcommand{\End}{{\rm End}} 
\newcommand{\Hom}{{\rm Hom}}
\newcommand{\lollipop}{\mbox{$\circ \kern-0.4em \rightarrow$}}

\newcommand{\dbdownarrow}{\rlap{\raise.25ex\hbox{$\shortdownarrow$}}\raise-.25ex\hbox{$\shortdownarrow$}}
\newcommand{\dda}{\rlap{\raise.25ex\hbox{$\shortdownarrow$}}\raise-.25ex\hbox{$\shortdownarrow$}}
\newcommand{\dua}{\rlap{\raise-.25ex\hbox{$\shortuparrow$}}\raise.25ex\hbox{$\shortuparrow$}}
\newcommand{\dbuparrow}{\rlap{\raise-.25ex\hbox{$\shortuparrow$}}\raise.25ex\hbox{$\shortuparrow$}}
\newcommand{\funion}{\mathrel{\makebox[0pt][l]{\hspace{.08em}\raisebox{.4ex}{\rule{.5em}{.1ex}}}\mathord{\cup}}}

\newcommand{\dsup}{\mathop{\bigvee{}^{^{\,\makebox[0pt]{$\scriptstyle\uparrow$}}}}}

\newcommand{\diamondplus}{\mathop{\Diamond\mkern-13.9mu\raise.22ex\hbox{$+$}}}
\newcommand{\diamonddot}{\mathop{\Diamond\mkern-9.5mu\raise.2ex\hbox{$\cdot$}}\,}

\renewcommand{\diamonddot}{\cdot_{_{_{\mbox{}\!\!\!\!S}}}\kern-.1em}
\renewcommand{\diamondplus}{+_{_{_{\mbox{}\!\!\!S}}}\kern-.1em}
\renewcommand{\odot}{\cdot_{_{_{\mbox{}\!\!\!\!P}}}\kern-.2em}
\renewcommand{\boxdot}{\cdot_{_{_{\mbox{}\!\!\!\!H}}}\kern-.2em}
\renewcommand{\boxplus}{+_{_{_{\mbox{}\!\!\!H}}}\kern-.2em}

\newtheorem{proposition}{Proposition}[section]
\newtheorem{corollary}{Corollary}[section]
\newtheorem{lemma}{Lemma}[section]
\newtheorem{remark}{Remark}[section]

\newcommand{\R}{\overline{\mathbb{R}}_{\mbox{\tiny +}}}
\newcommand{\oRp}{\R}
\newcommand{\RP}{{\mathbb{R}}_{\mbox{\tiny +}}}
\newcommand{\Rp}{\RP}

\DeclareMathOperator{\da}{\downarrow\!}

\newtheorem{definition}{Definition}[section]

\newtheorem{example}{Example}[section]

\title{On the equivalence of state transformer semantics and predicate
  transformer semantics\\[1cm] {\large Dedicated to Viktor
    Selivanov\\[-3mm] at the occasion of his 6o$^{th}$ birthday}}
\author{
Klaus Keimel\thanks{Fachbereich Mathematik, Technische Universit\"at
  Darmstadt, 64342 Darmstadt, Germany,
  email: {\tt keimel@mathematik.tu-darmstadt.de}\newline Work
  supported by  Deutsche Forschungsgemeinschaft
  (DFG).}
}

\begin{document}

\maketitle


G. Plotkin and the author \cite{KP} have worked out the equivalence
between state transformer semantics and predicate transformer
semantics in a domain theoretical setting for programs combining
nondeterminism and probability. Works of C. Morgan and co-authors
\cite{MM}, Keimel, Rosenbusch and Streicher \cite{KRS,KRS1},
already go in the same direction using only discrete state spaces. 
In fact, Keimel and Plotkin did not restrict to probabilities or
subprobabilities, 
but worked in an extended setting admitting positive measures that may
even have infinite values. This extended setting offers technical
advantages. It was the intention of the authors to cut down their
results to the subprobabilistic case in a subsequent paper. A paper by
J.~Goubault-Larrecq \cite{GL} already goes in this direction. When
preparing a first version of the follow-up paper, the
author of this paper wanted to clarify for himself the basic ideas. In fact,
the paper \cite{KP} is technically quite involved, and when one
reaches the last section, where the equivalence of predicate and state
transformer semantics is finally put together, one is quite
exhausted and has difficulties to see the leading ideas. Even the
referee of the paper seemed to have given up at that point.       

It is the aim of this paper to begin from the other end. In all the
situations that the author has been dealing with, the state
transformer semantics
had been given by a monad $\cT$ over the category {\sf DCPO} of
directed complete posets (= dcpos) and Scott-continuous functions (=
functions preserving the partial order and suprema of directed
subsets). A state transformer interprets the input-output behavior of
a program by a
Scott-continuous map $t$ from the input domain $X$ to the
'powerdomain' $\cT Y$ over the output domain $Y$. Thus, state
transformers live in the Kleisli category associated with the monad
$\cT$. If there is an equivalent predicate transformer semantics,
predicate transformers have to live in a category (dually) equivalent
to the Kleisli category.

In my experience the
equivalence between state and predicate transformer semantics is based
on a very simple principle derived from the continuation monad. One
starts with a dcpo $R$ of 'observations'. The elements of the function
space (the exponential) $R^X$ are 'observable predicates' over the
dcpo $X$, and maps $s\colon R^Y\to R^X$ are
'predicate transformers'.  
Assigning to every dcpo $X$ the space $R^{R^X}$ of maps $\varphi\colon
R^X\to R$ gives rise to a monad, the 'continuation monad'. The maps
$t\colon X\to R^{R^X}$ are 'state transformers'. It is a simple
observation that there is a natural bijection between state
transformers and predicate transformers (see Section \ref{sec:cont}).

Monads are used in denotational semantics to model computational
effects. In lots of cases they are obtained by using a dcpo $R$ of
observations carrying an additional algebraic structure. This
algebraic structure carries over to the function spaces $R^X$ and
$R^{R^X}$. It leads to
two kinds of monads 'subordinate' to the continuation monad. 
One may assign to each dcpo $X$ firstly the dcpo $\cM_R X\subseteq 
R^{R^X}$ of all
Scott-continuous algebra homomorphism $\varphi\colon R^X\to R$ (see
Section \ref{sec:transformers}) and
secondly the directed complete subalgebra $\cF_R X$ of  $R^{R^X}$
generated by the projections $\wh x = (f\mapsto f(x))\colon R^X\to R,
x\in X$ (see Section \ref{sec:free}).  

The first monad behaves nice with respect to the natural bijection
between state and predicate transformers:
The state transformers $t\colon X\to \cM_R X$ correspond to those
predicate transformers that are algebra homomorphisms $s\colon R^Y\to
R^X$ (see Section \ref{sec:transformers}). But this monad is, in general,
uninteresting for semantics. For semantics one uses the second monad
$\cF_R X$; but it is not clear to me how to characterize the predicate
transformers corresponding to the state transformers $t\colon X\to
\cF_R Y$. The situation becomes nice when both monads agree.\\

{\bf Problem.} Characterize those dcpo algebras $R$ for
which the monads $\cF_R$ and $\cM_R$ agree. \\

I do not have such a characterization. But I give a sufficient
condition for the containment $\cF_R X\subseteq \cM_R X$; this is the
case provided the algebraic structure on $R$ is 'entropic' (see
Section \ref{sec:entropic}). This
concept borrowed from universal algebra (see, e.g., the monograph
\cite{RS}) corresponds to commutativity for monads; but I have not
pursued this link. 

In Section \ref{sec:nondet} we deal with examples for the entropic
situation, powerdomains for nondeterminism (both angelic and demonic)
and for (extended) 
probabilistic choice. Our general approach yields the
containment relation   $\cF_R X\subseteq \cM_R X$ and the equivalence
between state and predicate transformers. The equality $\cF_R
X= \cM_R X$ has to be proved separately in each special
case. In fact, equality does not always hold: often one has to restrict
to continuous dcpos $X$. The proof of equality $\cF_R X= \cM_R X$
often is the really hard work, and for this we refer to the literature.    

The entropic condition does not cover situations combining  
nondeterministic and probabilistic features (see
\cite{KP}). Surprisingly there is a relaxed notion of entropicity (see
Section \ref{sec:laxentropic}) that allows to capture these situations
(see Section \ref{sec:combining}). The relaxation consists in
replacing certain equalities by 
inequalities. \\

{\bf Question.} Is there a concept for monads over {\sf DCPO} that
corresponds to this relaxed notion of entropicity?\\

In most presentations, powerdomains are described by collections of
certain subsets, for example (convex) Scott-closed sets, (convex)
Scott-compact saturated sets, lenses, etc. In this paper the advantage of
functional representations is put in evidence. The functional
representations may seem less intuitive. But a lot of
features become easier to prove and less technical to be handled. In
fact, in \cite{KP}, 
the set-based representations had to be translated into functional
representations in order to prove the equivalence of state and
predicate transformer semantics. 

There is a long tradition for using functional representations in
mathematics as well as in semantics:

-- The notion of a 'distribution' has been formalized by
L. Schwartz as a linear functional on the space of smooth functions
with compact support on $\mathbb R^n$. Here $\mathbb R$ corresponds to
the space of observations and the smooth functions are the 'predicates'.

-- Historically measures have first been introduced as certain set
functions on $\sigma$-algebras and integrals where defined as a
derived concept (see H. Lebesgue \cite{Les}). But alternatively 
the opposite approach has also been 
pursued: In the Daniell-Stone \cite{Dan,Sto} approach one starts with the
abstract notion of an integral, a positive linear 
functionals on a certain function space, and measures are obtained as 
derived notions. Bourbaki \cite{Bou} takes the same approach for measures on
topological spaces. 

-- In statistics and decision theory (see Walley \cite{Wal})
one uses the notion of a 'prevision' in the sense of our
'predicates'. Probabilities and upper/lower probabilities then arise from
functionals with certain properties on sets of previsions.

-- C. Morgan and co-authors (see \cite{MM}) use the notion of an
'expectation' in the sense of our 'predicates' in their
investigations on combining nondeterminism and probability. This
terminology is a bit misleading; 'prevision' or 'random variable' (as
proposed by D. Kozen \cite{Koz}) seems more
appropriate. Indeed in probability theory the term 'expectation'
denotes the mean value of a random variable. 

-- A. Simpson \cite{Sim} stresses the functional approach in
topological domain theory. For the angelic and demonic powerspace
constructions this approach has been carried through by Battenfeld and
Schr\"oder \cite{BS}. 

-- Labelled Markov processes, simulation and bisimulation are treated
in a very appealing way in recent work by Chaput, Danos, Panangaden
and Plotkin \cite{CDPP} using a functional approach and a duality that
reminds the 
equivalence between state and predicate transformers. Previously, these
constructions needed sophisticated tools from measure theory on Polish
and analytic spaces.\\

In this paper we do not present any new particular case. It is our aim
to present a general framework in which one can hope for a canonical
equivalence between state and predicate transformer semantics. I do
not have a proof, but I conjecture that under fairly general
hypotheses the framework presented in this paper is the only one in
which such an equivalence may happen:\\

{\bf Claim.} If a monad $\cT$ over the category {\sf DCPO} allows a
(dually) equivalent predicate transformer semantics, then there is a
dcpo $R$ with some additional relaxed entropic structure with the
following properties:   
$\cT X$ 'is' a substructure of $R^{R^X}$ consisting of all structure
preserving maps. $\cT X$ agrees with the substructure of $R^{R^X}$
generated by the projections $\wh x, x\in X$. The structure preserving
predicate transformers $s\colon R^Y\to R^X$ correspond naturally to the
state transformers $t\colon X\to\cT Y$.\\
   
Monads yield the free objects for the class of their Eilenberg-Moore
algebras. In all the examples that we look at in this paper the
powerdomain monads
are the free algebras for an (in)equational theory which reflects
properties of the choice operators involved. But here we do not elaborate
this topic. It is natural to conjecture that the monad $\cF_R$ yields
the free algebras for the (in)equational theory of the algebra $R$.\\

We use some basic background material from universal algebra, from
category theory and from
domain theory. We refer to Birkhoff's monograph \cite{Bir} for the
background on universal algebra, to Mac Lanes book \cite{ML} for
monads and Kleisli triples and to \cite{Dom} for directed complete
posets (dcpos) and continuous domains.

\section{Continuation monads and predicate transformers}
\label{sec:cont}

We will work in the category {\sf DCPO} of directed complete partially
ordered sets (dcpos) and Scott-continuous functions (maps preserving
the partial order and suprema of directed sets).

The category {\sf DCPO} is Cartesian closed. Finite products, even arbitrary 
products, are Cartesian products with the pointwise order; suprema of
directed sets are formed pointwise. The exponential consists
of all Scott-continuous functions $u\colon 
X\to Y$ with the pointwise defined order; suprema of directed sets of
Scott-continuous functions are
also formed pointwise. We use two notations for the exponential of $X$
and $Y$ in parallel: $$Y^X\ \ \ \  =\ \ \ \ [X\to Y].$$
The category {\sf SET} of sets and functions can be considered as a
full subcategory of {\sf DCPO}; just take the discrete order (equality) on
every set. Products and exponentials in {\sf SET} are the same as in
{\sf DCPO}.

We will use notations from simply typed $\lambda$-calculus. Since we 
are in a Cartesian closed category, all maps defined through well-typed
$\lambda$-calculus expressions are Scott-continuous (see,
for example, \cite[Part I]{LS}). Thus, we never need to prove the
continuity of functions. 
In order to avoid explicit type information, we
will fix notations as follows:\\

\begin{tabular}{l}
$R$ will be a fixed dcpo, called the dcpo of 'observations';\\ 
$X$ and $Y$ denote arbitrary dcpos;\\
$x$, $y$ and $r$ denote elements of $X$, $Y$ and $R$, respectively;\\ 
$u$ denotes Scott-continuous maps $u\colon X\to Y$, that is, elements
of $Y^X=[X\to Y]$;\\
$f$ and $g$ denote Scott-continuous maps $f\colon X\to R$ and $g\colon
Y\to R$, that is, elements of $R^X$ and $R^Y$;\\ 
$\varphi$ and $\psi$ denote Scott-continuous maps $\varphi\colon R^X\to R$ and
$\psi\colon R^Y\to R$.
\end{tabular}{l}\\

For a Scott-continuous map $u\colon X\to Y$ we denote by $R^u$ the map
from $R^Y$ to $R^X$ defined by $R^u=\lambda g.\ g\circ u$, that is,
$R^u(g)=g\circ u$ for all $g\in R^Y$; being even more explicit,
$R^u(g)$ is the map 
from $X$ to $R$ defined by $R^u(g)(x)= g(u(x))$ for all $x\in
X$. Note that we use the bracketing convention that $R^u(g)(x)$ has
to be read as $\big(R^u(g)\big)(x)$, a convention that we will use throughout.
We obtain a contravariant functor
$R^-$ from the category {\sf DCPO} into itself. 

Applying the contravariant functor $R^-$ twice yields a covariant
functor $R^{R^-}=[R^-\to R]$. This is the well-known
\emph{continuation monad}.  
The unit $\delta$ of the monad is given by the \emph{projections} or
\emph{evaluation maps} $\delta_X\colon X\to[R^X\to R]$ defined by  
\begin{eqnarray}\label{equ:delta}
\delta_X &=& \lambda x.\lambda f.\ f(x), 
\mbox{ that is, } \delta_X(x)(f) = f(x)
\end{eqnarray}
for $f\colon X\to R$ and $x\in X$. It will be convenient to introduce
the short notation $\widehat x$ for $\delta_X(x)$; then the defining
identity simply reads $\widehat x(f) = f(x)$. 

For a map $t\colon X\to [R^Y\to R]$ its 
Kleisli lifting $t^\dagger\colon [R^X\to R]\to  [R^Y\to R]$ is given by
 \begin{eqnarray}\label{equ:lifting}
t^\dagger=\lambda\varphi.\ \lambda g.\ \varphi(\lambda x.\
 t(x)(g)), \mbox{ that is, } t^\dagger(\varphi)(g) = \varphi(\lambda
 x.\ t(x)(g)).
\end{eqnarray}

For a Scott-continuous map $u\colon X\to Y$, the map $R^{R^u}\colon
R^{R^X}\to R^{R^Y}$ 
may be considered as a special case of a Kleisli lifting, namely
$R^{R^u} =(\delta_Y\circ u)^\dagger$, that is, $$R^{R^u}(\varphi)(g)
= \varphi(\lambda x.\ (\delta_Y\circ u)(x))(g)) = \varphi(\lambda x.\
g(u(x)))=\varphi(g\circ u).$$

\begin{lemma}\label{lem:delta}
The map $\delta_X\colon X\to [R^X\to R]$ is
Scott-continuous; if $R$ has at least two elements $r<s$, then $\delta_X$ is an
order embedding.
\end{lemma}

\begin{proof}
As $\delta_X$ is defined by a $\lambda$-expression, it is  
Scott-continuous.
If $x\not\leq x'$, then there is a Scott-continuous function 
$f\colon X\to R$ such that $f(x')=r<s=f(x)$, whence
$\delta_X(x)\not\leq\delta_X(x')$. (Just define $f(y)=r$ if $y\leq
x'$, else $= s$.) It follows that $\delta_X$ is an order embedding.   
\end{proof}

In many situations, objects $X$ occur as 'state spaces'. An 'action', a
'program', acts on states and produces results $\psi\in [R^Y\to R]$ over a
possibly different state space $Y$. Scott-continuous maps $t\colon
X\to [R^Y\to R]$ will be called \emph{state transformers}. The
elements $g\in R^Y$ are interpreted as 
'predicates', 'previsions', etc. Scott-continuous maps $s\colon R^Y\to R^X$ are
called \emph{predicate transformers}. 

\begin{lemma}\label{lem:PQ}
The dcpos of state and
predicate transformers are canonically isomorphic: $[R^Y\to R]^X \cong
[R^Y\to R^X]$. 
The mutually inverse natural bijections $P$ and $Q$ are 
given as follows:
The map $P$ assigns to a state transformers $t\colon X\to [R^Y\to R]$ the
predicate 
transformer $P(t)\colon R^Y\to R^X$ defined by $P(t) =\lambda g.\
\lambda x.\ t(x)(g)$, 
that is, $$P =\lambda t.\ \lambda g.\ \lambda x.\ t(x)(g).$$ 
The other way around, we assign to every
predicate transformer $s\colon R^Y\to R^X$ the state transformer 
$Q(s) = (R^s)\circ \delta_X= (\lambda \varphi.\ \varphi\circ
s)\circ(\lambda x.\ \lambda f.\ f(x)) =\lambda x.\ (\lambda f.\
f(x))\circ s$, whence $Q(s)(x)=(\lambda f.\ f(x))\circ s$. 
Thus $Q(s)(x)(g) = (\lambda f.\ f(x))(s(g)) = s(g)(x)$,
that is, $$Q = \lambda s.\ \lambda x. \lambda g.\  s(g)(x).$$ 
\end{lemma}

\begin{proof}
$P$ and $Q$ are mutually inverse bijections between state transformers
and predicate transformers. Indeed, for all maps
$s\colon R^Y\to R^X$, $g\in R^Y$ and $x\in X$, we have 
$P(Q(s))(g)(x) =Q(s)(x)(g) = s(g)(x)$, that is $P(Q(s)) = s$ for all
$s$. Similarly, for all $t\colon
X\to [R^Y\to R]$, $g\in R^Y$ and all $x\in X$, we have 
$(Q(Pt))(x)(g) = P(t)(g)(x) = t(x)(g)$, that is $Q(P(t)) = t$.
\end{proof}

In this paper we will consider monads that arise 'inside' the 
continuation monad in the following sense:
We assign to every dcpo $X$ a sub-dcpo $\cT X$ of $[R^X\to X]$ in
such a way that the following properties are satisfied:
\[\begin{array}{rcll}
\delta_X(x)&\in& \cT X &\mbox{ for all } x\in X;\\
t^\dagger(\cT X)&\subseteq& \cT Y &\mbox{ for every Scott-continuous }
t\colon X\to[R^Y\to R].
\end{array}\]
Then $R^{R^u}$ maps $\cT X$ into $\cT Y$ for every Scott-continuous
map $u\colon X\to Y$; indeed, $R^{R^u}$ is the Kleisli lifting
$(\delta_Y\circ u)^\dagger$ of $\delta_Y\circ u\colon X\to [R^Y\to
R]$. Thus $R^{R^u}$ induces a Scott-continuous map $\cT u\colon \cT
X\to\cT Y$ in such a way that $\cT$ becomes a functor, and even a
monad with (the corestriction of) $\delta$ as unit and the Kleisli
lifting $t^\dagger|_{\cT X}$ for $t\colon X\to \cT Y$.
  
\begin{definition}\label{def:subordinate}{\rm
We say that $\cT$ is a monad \emph{subordinate} to the continuation
monad if it arises in the way just described.}
\end{definition}

\section{Adding algebraic structure}\label{sec:algebra}

We recall a few concepts from universal algebra.
Every algebra has a well defined signature which consists of
operations symbols of a prescribed arity which will be supposed to be
finite in this paper.

\begin{definition}\label{def:signature}{\rm
A \emph{signature} $\Omega$ is the disjoint union (sum) of a
sequence of sets $\Omega_n, n\in \N$. The members $\omega\in
\Omega_n$ are called \emph{operation symbols of arity} $n$.}
\end{definition}

In our examples there will be no operation symbols of arity $n\geq
3$. Thus we can
restrict our attention to $\Omega_0,\Omega_1,\Omega_2$, that is, to
nullary, unary and binary operation symbols, and in many
cases there will be only finitely many operation symbols altogether. The
operation symbols of arity $0$ are also called constants.

\begin{definition}\label{def:algebra}{\rm
An \emph{algebra of signature} $\Omega$ consists of a set $A$ together with
operations $$\omega^A\colon A^n\to A,$$ one for each $\omega\in \Omega$
of arity $n$.
A map $u\colon A\to B$ from one algebra $A$ of signature
$\Omega$ to another one, $B$, is a \emph{homomorphism}, if   
$$u(\omega^A(a_1,\dots,a_n)= \omega^B(u(a_1),\dots,u(a_n))$$ for every
operation $\omega\in\Omega$ of arity $n$ and all
$a_1\dots,a_n\in A$.}
\end{definition}

We now replace the category of sets by the category {\sf DCPO} of
dcpos and Scott-continuous functions. 
We first adapt the notion of a signature $\Omega$ to the base category
{\sf DCPO}: 

\begin{definition}\label{def:d-signature}{\rm
A \emph{d-signature} $\Omega$ is the disjoint union (sum) of a sequence
of dcpos $\Omega_n$, $n\in \N$. 
}
\end{definition}

In all our examples, $\Omega_n$ will be empty for $n\geq
3$; $\Omega_0, \Omega_1, \Omega_2$ will consist of
finitely many operation symbols in most cases (trivially ordered),
and then a d-signature will be the same as a signature. But we will
consider some  
cases in which $\Omega_1$ will be a proper dcpo. By replacing the
dcpos $\Omega_n$ by their underlying sets $|\Omega_n|$ one retrieves a
signature as above.

\begin{definition}\label{def:d-algebra}{\rm
A \emph{directed complete partially ordered algebra} (a
\emph{d-algebra}, for short) of d-signature $\Omega$ is an algebra $A$
of signature $|\Omega|$ endowed with a structure of a dcpo in such a
way that the maps $$(\omega,(a_1,\dots,a_n))\mapsto
\omega^A(a_1,\dots,a_n) \colon  \Omega_n\times A^n\to A$$
are Scott-continuous for all $n$.  
 A Scott-continuous algebra homomorphism  between two
d-algebras of the same d-signature is shortly called \emph{d-homomorphism}.}
\end{definition}
 
Note that $\omega^A(a_1,\dots,a_n)$ depends continuously not only on
the $a_i$ but also on $\omega$ (which is a vacuous requirement, if
$\Omega_n$ consists of a single operation symbol or of finitely many
mutually incomparable operation symbols).

We fix a d-signature $\Omega$ for the rest of this paper and all
d-algebras are understood to be of this signature. 
We also fix a d-algebra $R$ of d-signature $\Omega$. 

For every dcpo $X$, the function space $R^X$ is also a d-algebra.  
For an  operation $\omega\in\Omega$ of arity $n$ the natural extension
of $\omega^R$ to a Scott-continuous operation on the 
function space $R^X$ is defined pointwise: For all $f_1,\dots,f_n\in
R^X$\display
\begin{eqnarray}
\omega^{R^X}(f_1,\dots,f_n)(x)=\omega^R(f_1(x),\dots,f_n(x))
\mbox{ for all } x\in X.\label{pointwiseoperation}
\end{eqnarray}

In the future, we will often omit the superscripts in $\omega^R,
\omega^{R^X}, ... $ and write simply $\omega$. This
simplification does not give rise to misunderstandings.

For every Scott-continuous map $u\colon X\to Y$, the induced
Scott-continuous map $R^u\colon R^Y\to R^X$ is a d-homomorphism: 
Indeed, $R^u(\omega(g_1,\dots,g_n))(x) =
\omega(g_1,\dots,g_n)(u(x))= \omega(g_1(u(x)),\dots,g_n(u(x)))= 
\omega(R^u(g_1)(x),\dots,R^u(g_n)(x)) =
\omega(R^u(g_1),\dots,R^u(g_n))(x)$ for all $x\in X$. Thus, we may
view $R^-$ to be contravariant functor from the category {\sf DCPO} to
the category of d-algebras and d-algebra homomorphisms.

In the same way, the operations $\omega$ can be extended to
$R^{R^X}=[R^X\to R]$ so that the latter becomes a d-algebra, too, and the
maps $R^{R^u}$ are d-algebra homomorphisms. Since it will be used
frequently, let us repeat the definition of the operations $\omega$ of
arity $n$ on $R^{R^X}$:  For all $\varphi_1,\dots,\varphi_n\in
R^X$: 
\begin{eqnarray}
\omega(\varphi_1,\dots,\varphi_n)(f)=\omega(\varphi_1(f),\dots,\varphi_n(f))
\mbox{ for all } f\in R^X.\label{pointwiseoperation1}
\end{eqnarray}

\begin{lemma}\label{lem:kleisli}
{\rm (a)} The projections $\delta_X(x)\colon R^X\to R$ are
d-homomorphisms for every $x\in X$.\\ 
For every state
transformer $t\colon X\to [R^Y\to R]$, the Kleisli lifting 
$t^\dagger\colon [R^X\to R]\to [R^Y\to R]$ 

{\rm (b)} is a d-homomorphism and 

{\rm (c)} maps d-homomorphisms $\varphi\colon R^X\to R$ to
d-homomorphisms $t^\dagger(\varphi)\colon R^Y\to R$.
\end{lemma} 

\begin{proof}
(a)  Let $x\in X$. For $\omega\in\Omega$
of arity $n$ and $f_1,\dots,f_n\in R^X$ we have
$\delta_X(x)(\omega(f_1,\dots,f_n)) = \omega(f_1,\dots,f_n)(x) =
\omega(f_1(x),\dots,f_n(x)) =
\omega(\delta_X(x)(f_1),\dots,\delta_X(x)(f_n))$.

(b) We have to check
that, for every $\omega\in\Omega_n$ and all
$\varphi_1,\dots,\varphi_n$, we have
$t^\dagger(\omega(\varphi_1,\dots,\varphi_n)) =
\omega(t^\dagger(\varphi_1),\dots,t^\dagger(\varphi_n))$. 
For every $g\in R^Y$ we have indeed:  
\[\begin{array}{rcll}
t^\dagger\big(\omega(\varphi_1,\dots,\varphi_n)\big)(g) &=&
\omega(\varphi_1,\dots.\varphi_n)(\lambda x.\ t(x)(g))&\mbox{by the
  definition (\ref{equ:lifting}) of $t^\dagger$} \\ 
 &=& \omega\big(\varphi_1 (\lambda x.\ t(x)(g)),
 \dots,\varphi_n(\lambda x.\ t(x)(g))\big)& \mbox{since $\omega$ is
   defined pointwise (\ref{pointwiseoperation1})}\\ &=&
 \omega\big(t^\dagger(\varphi_1)(g),\dots,t^\dagger(\varphi_n)(g)\big)&\mbox{by
   the definition (\ref{equ:lifting}) of $t^\dagger$} \\&=&
 \omega\big(t^\dagger(\varphi_1),\dots,t^\dagger(\varphi_n)\big)(g)&\mbox{since
   $\omega$ is defined pointwise (\ref{pointwiseoperation1}).} 
\end{array}\]

(c)  Let $\varphi\colon R^X\to R$ be a d-homomorphism. For arbitrary
$\omega\in\Omega_n$  
and arbitrary $g_1,\dots,g_n\in R^Y$, we have:
\[\begin{array}{rcll}
t^\dagger(\varphi)(\omega(g_1,\dots,g_n))&=&\varphi(\lambda x.\
t(x)(\omega(g_1,\dots,g_n)))&\mbox{by the definition of }t^\dagger \\
&=&\varphi(\lambda x.\ \omega(t(x)(g_1),\dots,t(x)(g_n)))& \mbox{since
  $t(x)$ is a homomorphism}\\ 
&=&\varphi(\omega(\lambda x.\ t(x)(g_1),\dots,\lambda x.\
t(x)(g_n)))&\mbox{since $\omega$ is defined pointwise}\\
&=&\omega(\varphi(\lambda x.\ t(x)(g_1)),\dots,\varphi(\lambda x.\
t(x)(g_n)))&\mbox{since $\varphi$ is a homomorphism}\\
&=&
\omega(t^\dagger(\varphi(g_1)),\dots,t^\dagger(\varphi(g_n)))&\mbox{by
  the definition of } t^\dagger.
\end{array}\]    
\end{proof}

\section{Monad I: Homomorphism monads}\label{sec:transformers}   

We continue with a fixed d-algebra $R$ of d-signature $\Omega$. 

A subset $C$ of a dcpo $X$ is called a \emph{sub-dcpo} if the supremum $\dsup_i
x_i$ of any directed family $(x_i)_i$ of elements in $C$ stays in
$C$.
For two d-algebras $A$ and $B$ of the same d-signature, we denote by
$$[A\lollipop B]$$ the set of all d-homomorphisms $u\colon A\to
B$. Since the pointwise supremum of a directed family of
d-homomorphisms is again a d-homomorphism, we have:

\begin{lemma}\label{lem:d-hom}
For any two d-algebras $A$ and $B$ of the same d-signature, the
d-homomorphisms $u\colon A\to B$ form a sub-dcpo 
$[A\lollipop B]$ of the dcpo $[A\to B]$ of all Scott-continuous maps from
$A$ to $B$. 
\end{lemma}  

In particular $[R^X\lollipop R]$, the set of all d-homomorphisms
$\varphi\colon R^X\to R$, is a sub-dcpo of 
$[R^X\to R]$ by Lemma \ref{lem:d-hom}. 
 By Lemma \ref{lem:kleisli}(a), $\delta_X(x)\in[R^X\lollipop R]$ for
 all $x\in X$. By  Lemma \ref{lem:kleisli}(c), for every state
transformer $t\colon X\to[R^Y\to X]$, its Kleisli lifting maps
$[R^X\lollipop R]$ into  $[R^Y\lollipop R]$.   
We are in the situation described in Definition
\ref{def:subordinate}): 

\begin{proposition}\label{prop:monad}
For a d-algebra $R$, the assignment $X\mapsto [R^X\lollipop R]$
yields a monad subordinate to the continuation monad. The unit is (the
corestriction of $\delta$ and the Kleisli lifting of a
Scott-continuous map $t\colon X\to [R^Y\lollipop R]$ is (the
restriction-corestriction) $t^\dagger \colon[R^X\lollipop
R]\to[R^Y\lollipop R]$
\end{proposition}

The 'homomorphism monad' $([R^-\lollipop R],\delta,{}^\dagger)$ exhibited
in the previous proposition has the remarkable property that it
behaves well with respect to the one-to-one correspondence between 
state and predicate transformers in the sense that there is a simple
characterization of the predicate transformers corresponding to the
state transformers $t\colon X\to[R^X\lollipop R]$\display

\begin{proposition}\label{prop:P-Q}
Let $R$ be d-algebra.
The maps $P$ and $Q$ (see Lemma {\rm \ref{lem:PQ}}) induce a one-to-one
correspondence between predicate transformers  
$s\colon R^Y\to R^X$ that are d-homomorphisms and those state
transformers $t\colon X\to [R^Y\to R]$ for which each $t(x), x\in X,$ is a
d-homomorphism: $$[R^Y\lollipop
R]^X\cong [R^Y\lollipop R^X]$$ 
\end{proposition}

\begin{proof}
Let $t(x)$ be a d-algebra homomorphism. For all $\omega\in\Omega$ of
arity $n$ and for all $g_1,\dots,g_n\in R^Y$ we have\display  
\[\begin{array}{rcll}
P(t)(\omega(g_1,\dots,g_n))(x) &=& t(x)(\omega(g_1,\dots,g_n)) &
\mbox{by Lemma \ref{lem:PQ}}\\ 
 &=& \omega(t(x)(g_1),\dots,t(x)(g_n))) & \mbox{since $t(x)$ is a
   homomorphism}\\
& =&\omega(P(t)(g_1)(x),\dots,P(t)(g_n)(x)) &\mbox{again by Lemma
  \ref{lem:PQ}}\\ 
&=&\omega\big(P(t)(g_1),\dots,P(t)(g_n)\big)(x)&\mbox{since $\omega$ is 
defined pointwise on $R^X$}
\end{array}\] 
If this holds for all $x\in X$, then
$P(t)\big(\omega(g_1,\dots,g_n)\big)
=\omega\big(P(t)(g_1),\dots,P(t)(g_n)\big)$ which 
shows that $P(t)$ is a homomorphism. If conversely $s\colon
R^Y\to R^X$ is a d-homomorphism, then
$Q(s)(x)(\omega(g_1,\dots,g_n))= s(\omega(g_1,\dots,g_n)(x)) =
\omega(s(g_1),\dots,s(g_n))(x) = \omega(s(g_1)(x),\dots,s(g_n)(x)) =
\omega(Q(s)(x)(g_1),\dots,Q(s)(x)(g_n))$ which shows that $Q(s)(x)$ is
a homomorphism. 
\end{proof}

For later use let us record the following: Let us replace the dcpo $X$
by a d-algebra $A$. The map $\delta_A\colon A\to [R^A\to R]$ is by no
means a homomorphism. But after replacing $R^A$ by $[A\lollipop R$ the
situation changes: For $a\in A$, the map $\delta_A(a)$ from $R^A$ to $R$
is restricted to a map from $[A\lollipop R]$ to $R$; we still use the
same notation  $\delta_A(a)$ for the restricted map.  

\begin{lemma}\label{lem:unit}
For a d-algebra $A$, the unit $\delta_A\colon A\to R^{[A\lollipop R]}$
is a d-homomorphism. 
\end{lemma}

\begin{proof}
We just have to show that $\delta_A$ is a homomorphism, that is, for
every operation $\omega$ of arity $n$ and all $a_1,\dots,a_n\in A$, we
have
$\delta_A(\omega(a_1,\dots,a_n))=\omega(\delta_A(a_1),\dots,\delta_A(a_n))$. 
Indeed,
for every d-homomorphism $h\colon A\to R$ we have
$\delta_A(\omega(a_1,\dots,a_n))(h) =
h(\omega(a_1,\dots,a_n)) =
\omega(h(a_1),\dots,h(a_n)) = 
\omega(\delta_A(a_1)(h),\dots,\delta_A(a_n)(h)) = 
\omega(\delta_A(a_1),\dots,\delta_A(a_n))(h)$.
\end{proof}

\section{Monad II: Free algebras}\label{sec:free}

We keep the setting of the previous section and consider a fixed
d-algebra $R$ of d-signature $\Omega$. In the previous section we
exhibited a monad $[R^-\lollipop R]$ over the category of dcpos
subordinate to the continuation monad by restricting to d-homomorphisms
$R^X\to R$ instead of arbitrary Scott-continuous maps. I doubt
that this monad is of any intrinsic 
interest. Of real interest for semantics and otherwise are free
algebras. We consider a second monad subordinate to the continuation
monad and we will investigate in which sense this is a free construction.  

 The intersection of any family of sub-dcpos in a dcpo $X$ is a sub-dcpo, in
fact, the sub-dcpos are the closed sets of a topology, called the
\emph{d-topology} (see, e.g.,\cite[Section 5]{KL1}). 

A subalgebra of a d-algebra $A$ which is a sub-dcpo, too, is called a
\emph{d-subalgebra}. The intersection of any family of d-subalgebras is
  again a d-subalgebra. 

\begin{lemma}\label{lem:d-closure}{\rm \cite[Corollary 5.7]{KL2}}
In any d-algebra $A$, the d-closure of a subalgebra $B$, that is, the
smallest sub-dcpo containing $B$, is a
d-subalgebra.  
\end{lemma}

For every dcpo $X$ we consider the
d-subalgebra  $\cF_R X$ of  $[R^X\to R]$ generated by the projections
$\widehat x=\delta_X(x), x\in X$. Indeed, since the intersection 
of any family of d-subalgebras is again a d-subalgebra,     
there is a smallest d-subalgebra $\cF_R X$ in $[R^X\to R]$ containing the
projections $\widehat x$, $x\in X$.

For a map $t\colon X\to \cF_R Y\subseteq [R^Y\to R]$, the 
Kleisli lifting $t^\dagger\colon [R^X\to R]\to [R^Y\to R]$ maps $\cF_R
X$ into $\cF_R Y$, since $t^\dagger$ is a
d-homomorphism by Lemma \ref{lem:kleisli}(b). This shows that
$(\cF_R,\delta,{}^\dagger)$ is a monad subordinate to the continuation
monad in the sense of Definition \ref{def:subordinate}, the Kleisli 
lifting of a map $t\colon X \to\cF_R Y$ being the restriction and
corestriction of the Kleisli 
lifting $\dagger$ for the continuation monad $[R^-\to R]$.

\begin{proposition}\label{prop:free}
$(\cF_R,\delta,{}^\dagger)$ is a monad over the category {\sf DCPO}
subordinate to the continuation monad.
\end{proposition} 

Since we have a monad, the d-algebras $\cF_R X$ are free for the class
of its Eilenberg-Moore algebras. It is a challenge to determine these
Eilenberg-Moore algebras more concretely. The natural conjecture is
that the $\cF_R X$ are free over $X$ for the class of d-algebras
determined by the (in)equational theory of the d-algebra $R$. But in this
paper we will not discuss this question.

Of course, we can consider the predicate transformers $s\colon R^Y\to
R^X$ that correspond to state transformers $t\colon X\to \cF_R Y$ under
the mutually inverse bijections $P$ and $Q$ according to Lemma
\ref{lem:PQ}. But I do not know of any intrinsic characterization of
these state transformers. This is in contrast to the situation that we
encountered in the previous section with the homomorphism
monad. Thus, with respect to a characterization of the predicate
transformers we are in an excellent
position in those cases where the monads $\cF_R$ and
$[R^-\lollipop R]$ agree. It would already be an advantage to be in the
position, where  the free d-algebra $\cF R X$ is contained in the dcpo
$[R^X\lollipop R]$. 

Clearly, the generators of $\cF_R X$, the projections $\widehat x,
x\in X,$ are homomorphisms and thus belong to $[R^X\lollipop R]$. But
no other element of $\cF_R X$ need to be a d-algebra homomorphism. 
For example, if we choose for
$R$ the d-semiring $\oRp$ (with two constants $0$ and $1$
and the two binary operations addition and multiplication) and for $X$
the unordered two element set,
then the two projections $(x_1,x_2)\mapsto x_i, i=1,2,$ are the only
Scott-continuous homomorphisms from 
$\oRp^2$ to $\oRp$. But the free d-algebra with two generators is
quite big, containing for example all polynomials in two variables
$x_i,x_2$ with nonnegative integer coefficients.

But we observe:

\begin{remark}\label{rem:subalg}
If the d-homomorphisms $\varphi\colon R^X\to R$ form a subalgebra of
$[R^X\to R]$, then $\cF_R X\subseteq [R^X\lollipop R]$ for every dcpo $X$.
{\rm Indeed, if $[R^X\lollipop R]$ is a subalgebra, then it is a
  d-subalgebra of $[R^X\to R]$. As the projections $\widehat x, x\in
  X,$ are homomorphisms, they belong to $[R^X\lollipop
  R]$. Hence, $[R^X\lollipop R]$ contains the d-subalgebra $\cF_RX$
  generated by the projections.  }
\end{remark}

We are led to ask the question under which hypothesis the
d-homomorphisms $\varphi\colon R^X\to R$ form a subalgebra of $[R^X\to
R]$. Classical universal algebra offers an answer to that
question.

\section{Entropic algebras}\label{sec:entropic}

Let us begin with classical universal algebra (over
the category of sets) and consider algebras $B$ and $R$ of the same
signature $\Omega$.  The 
set Hom$(B,R)$ of all algebra homomorphisms $\varphi\colon B\to R$ is
a subset of the product algebra $R^B$. We ask the
question, whether Hom$(B,R)$ is a subalgebra of $R^B$.  

In order to answer this question, consider an
operation $\sigma\in \Omega$ of arity $n$. For 
$\Hom(B,R)$ to be a sub-algebra we have to show that, for 
all $\varphi_1,\dots,\varphi_n\in \Hom(B,R)$, also
$\sigma(\varphi_1,\dots,\varphi_n)$ is an algebra 
homomorphism, that is, for every operation $\omega\in\Omega$ of arity $m$ and
for all $f_1,\dots,f_m\in B$, we have  
\begin{eqnarray}\label{homo}
\sigma(\varphi_1,\dots,\varphi_n)(\omega(f_1,\dots,f_m)) =  
\omega(\sigma(\varphi_1,\dots,\varphi_n)(f_1),\ \dots\ ,
\sigma(\varphi_1,\dots,\varphi_n)(f_m)).
\end{eqnarray}

\begin{definition}\label{def:commutes}
We will say that an operation $\sigma$ of arity $n$ and an operation
$\omega$ of arity $m$ on an algebra $R$ \emph{commute}
if, for all $x_{ij}\in R,\ i=1,\dots n,\ j=1,\dots,m$\display
\begin{eqnarray}\label{distributes}
\sigma(\omega(x_{11},\dots,x_{1m}),\ \dots\ ,\omega(x_{n1},\dots,x_{nm})) =  
\omega(\sigma(x_{11},\dots,x_{n1}),\ \dots\ ,
\sigma(x_{1m},\dots,x_{nm})).
\end{eqnarray}
\end{definition}

Such an equational law is also called an \emph{entropic law}. It can also be
expressed by the commutativity of the following diagram\display
\begin{diagram}
(A^m)^n\cong(A^{n})^{m}& \rTo^{\sigma^m} & A^m\\
\dTo<{\omega^n}&               & \dTo>\omega\\
A^n          &\rTo_\sigma      &A
\end{diagram}
If this entropic law holds in
$R$, it also holds in any power $R^I$, in $R^{R^I}$ and in all subalgebras
thereof. As a consequence,  
equation (\ref{homo}) holds if $\sigma$ commutes with $\omega$ in $R$. 
Indeed, 
\[\begin{array}{rcl}
\sigma(\varphi_1,\dots,\varphi_n)(\omega(f_1,\dots,f_m)) &= &
\sigma\big(\varphi_1(\omega(f_1,\dots,f_m)),\ \dots\
,\varphi_n(\omega(f_1,\dots,f_m))\big)\\ &&\mbox{(the operation
  $\sigma$ being defined pointwise)}\\ 
 &=& \sigma\big(\omega(\varphi_1(f_1),\dots,\varphi_1(f_m)),\ 
\dots\ ,\omega(\varphi_n(f(1),\dot,\varphi_n(f_m))\big)\\&&\mbox{(the
  $\varphi_i$ being homomorphisms)}\\ 
&=&\omega\big(\sigma(\varphi_1(f_1),\dots,\varphi_n(f_1)),\ \dots\
,\sigma(\varphi_1(f_m),\dots,\varphi_n(f_m))\big)\\ 
&&\mbox{(since $\sigma$ commutes with $\omega$ in $R$, equation
  (\ref{distributes}))} 
\end{array}\]

\begin{definition}\label{def:entropic}
An algebra of signature $\Omega$ is called \emph{entropic} if any two
operations $\sigma,\omega\in\Omega$ commute.
\end{definition}
 We have to be careful with
the nullary operations: If $c$ is a constant, then the entropic law
says that $\omega(c,\dots,c)=c$ and that two constants have to
agree. Thus, for an entropic algebra, we can suppose that there is at
most one nullary operation $c$ and, if there is one, the constant $c$
is a subalgebra, in fact, the smallest subalgebra. 

\begin{example}\label{ex:commutes}{\rm
(a) We have to be careful with
 nullary operations: If $c$ is a constant, then the entropic law
says that $\omega(c,\dots,c)=c$ and that two constants have to
agree. Thus, for an entropic algebra, we can suppose that there is at
most one nullary operation $c$ and, if there is one, the constant $c$
is a subalgebra, in fact, the smallest subalgebra. 

(b) A unary operation $\rho$ commutes with a binary operation $+$ if 
$$\rho(x+y)=\rho(x)+\rho(y).$$

(c) A binary relation $*$ commutes with itself if \footnote{This
  law has been called the \emph{entropic law} by I.M.H.~Etherington,
  \emph{Groupoids with additive endomorphisms}, American Mathematical
  Monthly {\bf 65} (1958), pages 596--601. Etherington used this law
  in order to characterize those groupoids, in which the pointwise
  product of two endomorphisms is again an endomorphism. For
  quasi-groups, the entropic law had already been considered under
  another name by D. C. Murdoch, \emph{Quasi-groups which satisfy
    certain generalized associativity laws}, American Journal of
  Mathematics {\bf 61}  (1939), pages 508--522.  
O. Frink Jr., \emph{Symmetric and self-distributive systems}, American
Mathematical Monthly {\bf 62} (1955), pages 697--707, used the notion
\emph{symmetric} for a binary operation $+$ satisfying nothing but the
entropic law. He notes that all means (arithmetic, geometric,
harmonic) are symmetric, that barycentric operations are entropic. For
a symmetric groupoid $(S,+)$ he proves that the powerset with the
induced operation $+$ is symmetric, that the pointwise defined sum of
two endomorphisms is an endomorphism and that the endomorphisms form a
symmetric groupoid under this operation; further, if $\alpha,\beta$
are two commuting endomorphisms, then $x*y=\alpha(x)+\beta(y)$ is a
symmetric operation, too. A lot more on entropic algebras one can
find in the monograph \cite{RS} by A.B.~Romanowska and J.D.H.~Smith.} 
$$(x_1 *x_2)*(x_3* x_4) = (x_1*x_3)*(x_3*x_4).$$
 In particular, every commutative,
associative binary operation commutes with itself. Thus, commutative
semigroups, commutative monoids, commutative groups and semilattices
are entropic. 

(d) Two binary operation $+$ and $*$ commute iff 
$$(x_1* x_2)+(x_3* x_4) = (x_1+x_3)*(x_2+x_4)$$
As this identity does not hold for addition and multiplication in semirings
and rings, these are not entropic. Similarly lattices, even
distributive lattices are not entropic. }
\end{example}

From the considerations preceding the definition we have:

\begin{proposition}\label{prop:entropic0}{\rm (see, e.g.,
    \cite[Proposition 5.1]{RS})} 
If the algebra $R$ is entropic and $B$ any algebra of the same signature,
the algebra homomorphisms $\varphi\colon B\to R$ form a subalgebra of
the product algebra $R^B$.  
\end{proposition}

We now turn to a d-algebra $R$ of d-signature $\Omega$. In this
section, all d-algebras are supposed to be of the same d-signature
$\Omega$. 

Since 
entropicity is defined by equational laws, every homomorphic image of a
subalgebra of a product of entropic algebras is entropic. Thus, if $R$
is an entropic d-algebra, then the function spaces $R^X$ and $R^{R^X}$
are entropic, too, as well as all d-subalgebras thereof. From the
previous proposition we deduce: 

\begin{corollary}\label{cor:entropic}
For an entropic d-algebra $R$ and any dcpo $X$, the collection
$[R^X\lollipop R]$ of all d-homomorphisms 
$\varphi\colon R^X\to R$ is a d-subalgebra of $[R^X\to R]$.
\end{corollary}

As a subalgebra of $[R^X\to R]$, the algebra $[R^X\lollipop R]$ is
again entropic. 

We now can state the first main result in this section. It follows
from the corollary above and Remark \ref{rem:subalg}:

\begin{proposition}\label{prop:free2} 
If $R$ is an entropic d-algebra, then $\cF_R X\subseteq [R^X\lollipop
R]$ for any dcpo $X$.
\end{proposition}

We have seen in Proposition \ref{prop:P-Q} that the state transformers
$t\colon X\to [R^Y\lollipop R]$ correspond bijectively to the
predicate transformers 
$s\colon R^Y\to R^X$ which are d-homomorphisms. Is there a
characterization of those predicate transformers that correspond to
the state transformers $t\colon X\to \cF_R X$? Of course, this is not a
problem in case $\cF_R X = [R^X\lollipop R]$. 

I do not know a general criterion for the equality $\cF_R X =
[R^X\lollipop R]$  to hold, even in the entropic setting.  It does not
hold in general\display 

\begin{example}\label{ex:reals}{\rm
The nonnegative extended reals with the constant $0$
and the binary operation $+$ form an entropic d-monoid
$R=(\oRp,+,0)$. For any dcpo 
$X$, $\cF_R X$ is the d-closure of the 
set of all finite sums $\sum_{i=1}^n n_i\wh x_i$ where the $n_i$ range over
positive integers and $x_i\in X$. But for every $r\in \Rp$ and every d-monoid
homomorphism  $\varphi\colon \oRp^X\to \oRp$ also $r\varphi$ is a d-monoid
homomorphism. Thus $[R^X\lollipop R]$ contains all scalar
multiples $r\wh x$ where $r$ ranges over positive reals, hence, all
linear combinations $\sum_{i=1}^n r_i\wh x_i$, where the $r_i$ range
over positive 
reals. But clearly, if $r$ is not an integer, 
then $r\wh x$ is not a member of $\cF_R X$. }  
\end{example}

 If we still want the equality $\cF_R X =[R^X\lollipop R]$ in the
 previous example, we have at least to enrich 
$\cF_R$ by allowing multiplication with scalars $r\in \Rp$. In order to
do this in the general setting, we observe that the maps $x\mapsto
rx\colon\oRp\to \oRp$ for $r\in\oRp$ are precisely the d-endomorphisms
of the d-monoid $R=(\oRp,+,0)$. In the general situation, every
d-homomorphism 
$\varphi\in [R^X\lollipop R]$ composed with a d-endomorphism $\rho$ of $R$
yields again a d-homomorphism $\rho\circ\varphi\in[R^X\lollipop R]$. 

 The following corollary arises as a special case of Corollary
 \ref{cor:entropic}, where $X$ consists of one element only:

\begin{corollary}\label{cor:entropic1}
For an entropic d-algebra $R$, the
Scott-continuous endomorphisms $\rho\colon R\to R$ form an entropic
d-algebra, the operations $\omega\in\Omega$ being defined pointwise. 
\end{corollary}

For the d-endomorphisms of $R$ we have an additional Scott-continuous
binary operation, the composition $\rho_1\circ \rho_2$. We denote by
$\End(R)\subseteq [R\to R]$ the d-algebra of all Scott-continuous
endomorphisms of $R$ 
with the operations $\omega\in \Omega$ and composition as an
additional binary operation. The identity map on $R$ is denoted by
$\mathbf 1_R$.  Thus, the d-signature of the d-algebra $\End(R)$ is
$\Omega$ augmented by composition and a constant for the identity.
Note that the operation $\circ$ destroys entropicity.

The d-algebra $\End(R)$ acts not
only on $R$ but also on $R^X$: For $f\in R^X$ and $\rho\in
$End$(R)$, $\rho\circ f$ is again an element of $R^X$. In a similar
way, End$(R)$ acts on $[R^X\to R]$. Composing a d-algebra homomorphism
$\varphi\in[R^X\lollipop R]$ with an endomorphism $\rho\in \End(R)$
yields again 
a d-algebra homomorphism $\rho\circ\varphi\in[R^X\lollipop R]$. Thus
$R$, $R^X$, $[R^X\to R]$ and $[R^X\lollipop R]$ are $\End(R)$-modules in
the sense of the following definition\display

\begin{definition}\label{def:modul}{\rm
Let $R$ be an entropic d-algebra of d-signature $\Omega$. An
$\End(R)$-\emph{d-module} is a d-algebra $A$ of d-signature $\Omega$
together with a Scott-continuous map $(\rho,x)\mapsto \rho\cdot
x\colon \End(R)\times A\to A$ satisfying the following axioms for all
$\rho,\rho_1,\dots,\rho_n\in \End(R)$, all $\omega\in\Omega$\ and all
$x,x_1,\dots,x_n\in A$\display 
\begin{eqnarray}
\mathbf 1_R\cdot x& =& x\label{ax1}\\
(\rho_1\circ \rho_2)\cdot x&=& \rho_1\cdot(\rho_2\cdot x)\label{ax2}\\
\omega(\rho_1,\dots,\rho_n)\cdot x &=& \omega(\rho_1\cdot
x,\dots,\rho_n\cdot x)\label{ax3} \\
\rho\cdot\omega(x_1,\dots,x_n)&=&\omega(\rho\cdot x_1,\dots,\rho\cdot
x_n)\label{ax4}
\end{eqnarray}
A map $u$ from an $\End(R)$-\emph{d-module} $A$ to an
$\End(R)$-\emph{d-module} $A'$ is said to be an
$\End(R)$-\emph{d-module homomorphism}, if
\begin{eqnarray}
u(\omega(x_1,\dots,x_n))&=& \omega(u(x_1),\dots,u(x_n))\ \ \mbox{ for all
}\omega\in\Omega\\
u(\rho\cdot x) &=&\rho\cdot f(x)\ \ \mbox{ for all } \rho\in \End(R).
\end{eqnarray} }
\end{definition}  
Axiom (\ref{ax4}) says that $\rho\mapsto\rho\cdot x$ is an $\Omega$-algebra
homomorphism from $\End(R)$ into $A$ for every fixed $x\in A$, and 
equation (\ref{ax3}) says that $x\mapsto\rho\cdot x$ is an 
endomorphism of $A$ for every fixed $\rho$. We can
subsume these two statement 
under the slogan that  $(\rho,x)\mapsto \rho\cdot
x\colon \End(R)\times A\to A$ is an $\Omega$-\emph{bimorphism}.

On an $\End(R)$-d-module $A$, we may interpret each endomorphism $\rho$
to be a unary operation on $A$. In this way, $A$ becomes a d-algebra
of d-signature $\Omega\cup \End(R)$. The defining axioms (\ref{ax1})
-- (\ref{ax4}) 
become equational laws. Axiom (\ref{ax4}) shows that the unary operations
$\rho$ commute with the operations $\omega\in \Omega$. We have\display

\begin{proposition}\label{prop:entropic}
For an entropic d-algebra $R$ of d-signature $\Omega$, every
$\End(R)$-d-module $A$ is an entropic d-algebra of d-signature  $\Omega\cup
\End(R)$ provided that $A$ is entropic for the signature $\Omega$. 
\end{proposition}

We are now in a position to reinterpret the material of this section
in the following way: We fix an entropic d-algebra $R$ of d-signature
$\Omega$ and we regard it as an $\End(R)$-d-module of d-signature
$\Omega\cup\End(R)$. It stays entropic by Proposition
\ref{prop:entropic}. For any dcpo $X$, the function spaces $R^X$ and
$[R^X\to R]$ are entropic $\End(R)$-d-modules, too; the module operation
is given by $\rho\circ f$ and $\rho\circ\varphi$ for $\rho\in\End(R)$,
$f\in\R^X$ and $\varphi\in R^{R^X}$, respectively. The subset
$[R^X\lollipop_mod R]$ of 
all $\End(R)$-d-module homomorphisms $\varphi\colon R^X\to R$ is 
an $\End(R)$-d-submodule of $[R^X\to R]$ by Proposition
\ref{prop:free}. 

Note that
$[R^X\lollipop_{mod} R]$ might be properly smaller than the set
$[R^X\lollipop R]$  of all
d-algebra homomorphisms $\varphi\colon R^X\to R$. Now $\cF_{mod}X$
will be the $\End(R)$-d-submodule of $[R^X\to R]$ generated by the
projections $\wh x, x\in X$. These projections are not only d-algebra
but also $\End(R)$-d-module
homomorphisms, since $\wh x(\rho\circ f) =\rho(f(x)) =\rho(\wh x(f)) =
(\rho\circ\wh x)(f)$ for $\rho\in\End(R)$ and $f\in R^X$. Thus, $\wh
x\in [R^X\lollipop_{mod}R]$. It follows that the
$\End(R)$-d-submodule $\cF_{mod}X$ generated 
by the projections in $[R^X\to R]$ is contained in $[R^X\lollipop_{mod}R]$. 

Although the $\End(R)$-d-module $\cF_{mod}X$ is bigger than the
d-algebra $\cF_R X$ the question remains open whether
$\cF_{mod}X=[R^X\lollipop_{mod}R]$. Maybe that this question has to be
decided in every special case separately.

\section{Examples: Powerdomains}\label{sec:nondet}

We want to illustrate that some standard powerdomain constructions
(see e.g. M.B.~Smyth \cite{Sm}) fit
under the framework developed until now. Powerdomains are
used for interpreting programs involving nondeterministic or
probabilistic choice.

Our basic domain of observations is the two element dcpo $\mathbf 2=\{0,1\}$
with $0<1$. Here $1$ denotes termination of a program and is
observable, while $0$ 
denotes nontermination which is not observable. 

For a dcpo $X$, an
observable predicate will be a Scott-continuous map $p\colon X\to
\mathbf 2$, and $\mathbf 2^X$ will be the domain of observable
predicates. The Scott-continuous functions from a dcpo $X$ to $\mathbf
2$ are the 
characteristic functions of Scott-open subsets. Thus, the domain
$\mathbf 2^X$ of observable predicates
can be identified with the complete lattice $\cO X$ of Scott-open
subsets of $X$ ordered by inclusion. 

A predicate transformer will be a Scott-continuous map $s\colon\mathbf
2^Y\to \mathbf 2^X$ or, equivalently, a Scott-continuous map $s\colon
\cO Y\to\cO X$. A state transformer will be a Scott-continuous map
$t\colon X\to [\mathbf 2^Y\to\mathbf 2]$, equivalently, $t\colon X\to
\cO\cO Y$, where $\cO\cO Y$ denotes the complete lattice of all
Scott-open subsets of the complete lattice $\cO Y$. According to Lemma
\ref{lem:PQ}, state and predicate transformers are in a canonical
one-to-one correspondence.\\

\subsection{The deterministic case}
For deterministic programs the state transformers $t$ will be
Scott-continuous maps from the input domain to the output domain:
\begin{diagram}
X&\rTo &Y&\rTo^{\delta_Y}&[\mathbf 2^Y\to\mathbf 2]
\end{diagram}
We can reason about properties of such programs using the connectives
'and' and 'or' as usual: If we can observe each of the predicates $p$
and $q$, then we can also observe their conjunction $p\wedge q$ and
their disjunction $p\vee q$. Thus, we consider our two element dcpo
$\mathbf 2$ as a d-algebra with two binary operations $\wedge$ (= max)
and $\vee$ (= max) and we add $0$ and $1$ as constants. The algebra
$(\mathbf 2,\wedge,\vee,0,1)$ is not entropic, so that our previous
developments do not apply. Let us describe this situation:

On $\mathbf 2^X$ and $[\mathbf 2^X\to\mathbf 2]$ the 
operations $\wedge$ and $\vee$ are pointwise binary inf and sup, the
constants being 
interpreted by the constant functions $0$ and $1$. On $\cO X$ and
$\cO\cO X$, the operations $\wedge$ and $\vee$ are interpreted by
$\cap$ and $\cup$, 
the constant by the empty set and the whole space, respectively.
Since we have directed suprema in dcpos anyway, the algebraic
structure is that of a frame: We have arbitrary suprema and finite infima
connected by meet-distributivity. The dcpo $[\mathbf
2^X\lollipop\mathbf 2]$ of frame homomorphisms $\varphi\colon \mathbf
2^X\to\mathbf 2$ is the sobrification of $X^s$ of $X$. The state
transformers $t\colon X\to Y^s$ are in bijective correspondence with
the predicate transformers $s\colon\mathbf 2^Y\to\mathbf 2^X$ which are
frame homomorphisms according to Proposition \ref{prop:P-Q}. We see:

Only if the dcpo $Y$ is a sober space in its Scott topology, the state
transformers $t\colon X\to Y$ are in bijective correspondence with the
frame homomorphisms $h\colon \mathbf 2^Y\to\mathbf 2^X$. But notice
that the frame generated in $[\mathbf 2^X\to\mathbf 2]$ by the
projections $\wh x, x\in X,$ is much bigger than $X^s$.  \\

\subsection{The nondeterministic case}
We now suppose that we interpret programs that admit a
nondeterministic choice operator $\funion$. The effect is that a
program, if it terminates, may lead to several results.
There are two basic ways for interpreting such a choice operator, the
\emph{angelic} and the \emph{demonic} interpretation. In the first
case we are happy if at least one of the possible outcomes has the
desired property, in the second case we demand that all of the
possible outcomes have the desired property. This boils down to
interpret the nondeterministic choice operator on our domain $\mathbf
2$ of observations by the binary operation $\vee$ in the first case,
but by the binary operation $\wedge$ in the second case. Thus our algebras of
observations are $$\mathbf 2_{ang}= (\mathbf 2,\vee,0) \ \ \ \mbox{ and }
\ \ \ \mathbf 2_{dem}= (\mathbf 2,\wedge,1)$$ 
for angelic and demonic nondeterminism, respectively.

In both cases we have a semilattice with unit, hence an entropic
d-algebra according to Example \ref{ex:commutes}(c). Accordingly, 
$\mathbf 2_{ang}^X$ and $[\mathbf 2_{ang}^X\to\mathbf 2_{ang}]$ will be
unital $\vee$-semilattices and $\mathbf 2_{dem}^X$ and $[\mathbf
2_{dem}^X\to\mathbf 2_{dem}]$ are unital 
$\wedge$-semilattices. In the equivalent presentation through
Scott-open sets, these function spaces correspond to $\cO X$ and
$\cO\cO X$ with binary union 
and $\emptyset$ as a constant in the angelic case and binary
intersection and the whole space as a constant in the demonic case.

For any dcpo $X$ we define 
$$\cH X =[\mathbf 2_{ang}^X\lollipop\mathbf 2_{ang}]$$
to be the dcpo of all d-$\vee$-semilattice homomorphisms
$\varphi\colon \mathbf 2_{ang}^X\to\mathbf 2_{ang}$. By Corollary \ref{cor:entropic}, $\cH
X$ is a d-$\vee$-subsemilattice with a bottom element. It is called
the \emph{angelic} or \emph{lower} or \emph{Hoare powerdomain} over $X$.
We have the following well-known result (see
e.g. \cite{Sm}): 

\begin{proposition}\label{prop:angelic}
(a) The angelic powerdomain $\cH X$ can be identified with the
complete lattice of all Scott-closed subsets of
$X$.\footnote{Most authors exclude the empty set, the 
  bottom element, from the Hoare powerdomain; then the predicate
  transformers are just supposed to preserve unions of nonempty
  families of closed sets.}

(b) The predicate transformers corresponding to the state transformers
$t\colon X\to \cH Y$ are those Scott-continuous maps $s\colon \cO
Y\to\cO X$ preserving binary unions and $\emptyset$ (hence those maps
preserving arbitrary unions because of Scott continuity).  

(c) The unital join d-subsemilattice $\cF_{\mathbf 2_{ang}}X$ generated
by the projections $wh x, x\in X,$ equals $\cH X$. 
\end{proposition}

\begin{proof}
(a) For every Scott-closed subset $C$ of $X$, the open sets $U$
contained in $X\setminus C$ form a Scott-closed ideal of the lattice $\cO X$;
hence, the map defined by $\varphi(U) =0$, if $U\cap C=\emptyset$, else
$=1$, is a Scott-continuous unital $\vee$-semilattice homomorphism and
every such homomorphism is of this form.

(b) follows from Proposition \ref{prop:P-Q}.

(c) follows from the fact that a Scott closed subset $C$ is the union
of the collection of principal ideals $\da x, x\in C,$ and that these
principal ideals $\da x$ correspond to the projections $\widehat x$
under the correspondence given in (a). 
\end{proof}

We see that our general developments yield the claims (a) and (b) of
the previous proposition. For the claim (c), our general developments
only tell us that the unital d-$\vee$-semilattice generated by the
projections is contained in $\cH X$. For the equality we have to use
the special situation. 

For any dcpo $X$ we define 
$$\cS X =[\mathbf 2_{dem}^X\lollipop\mathbf 2_{dem}]$$
to be the dcpo of all d-$\wedge$-semilattice homomorphisms
$\varphi\colon \mathbf 2_{dem}^X\to\mathbf 2_{dem}$. By Corollary
\ref{cor:entropic}, $\cS 
X$ is a d-$\wedge$-subsemilattice with a top element. It is called
the \emph{demonic} or \emph{upper} or \emph{Smyth powerdomain} over $X$.
We have the following well-known result (see
e.g. \cite{Sm}): 

\begin{proposition}\label{prop:demonic}
(a) The demonic powerdomain $\cS X$ can be identified with the
$\cap$-semilattice 
 of all Scott-open filters of $\cO X$. If $X$ is sober
for its Scott topology,  $\cS X$ can be identified with the
$\cup$-semilattice of 
Scott-compact saturated subsets of $X$ ordered by reverse
inclusion. \footnote{We have included $\emptyset$ as the top element
  of $\cS X$. Most authors exclude the empty set from the Smyth
  powerdomain.}

(b) The predicate transformers corresponding to the state transformers
$t\colon X\to \cS Y$ are those Scott-continuous maps $s\colon \cO
Y\to\cO X$ preserving binary intersections and the top (hence those maps
preserving finite intersections).  

(c) If $X$ is a continuous dcpo, the unital d-$\wedge$-subsemilattice
$\cF_{\mathbf 2_{ang}}X$ generated 
by the projections $\wh x, x\in X,$ equals $\cS X$. 
\end{proposition}

\begin{proof}
(a) Clearly a map $\varphi\colon  \mathbf 2_{dem}^X\cong\cO
X\to\mathbf 2_{dem}$ 
is Scott-continuous and preserves finite meets if and only of
$\varphi^{-1}(1)$ is a Scott-open filter of the complete lattice $\cO
X$. Thus $\cS X$ can 
be identified with the collection of all Scott open filters of $\cO
X$ (including $\cO X$ as a filter). In a sober space, the Scott open
filters $\cF$ of $\cO X$ correspond bijectively to the Scott-compact
saturated sets, the bijection being given by $\cF\mapsto \bigcap
\cF$ (see e.g. \cite[Theorem II-1.20]{Dom}). 

(b) follows from Proposition \ref{prop:P-Q}.

(c) For a proof we refer to \cite[Theorem IV-8.10]{Dom}.
\end{proof}

We see that our general developments yield the claims (a) and (b) of
the previous proposition. Concerning (c), we cannot use any general
principle. In general the unital
d-$\wedge$-subsemilattice of $[\mathbf 2_{dem}^X\to\mathbf
2_{dem}]$ can be strictly smaller than $\cS X$. As often, one has to
restrict here to continuous dcpos, where one can use approximations
from way-below.

\subsection{The extended probabilistic powerdomain}
\label{subsec:prob}

In order to catch probabilistic choice in programming, some kind of
measure theory had to be introduced for domains. Measures take
non-negative real values and possibly the value $+\infty$. In defining
measures one needs addition of nonnegative extended reals and suprema
of increasing sequences.

Thus, let $\oRp$ denote the dcpo of nonnegative real numbers augmented by
$+\infty$ with the usual linear order. The algebraic structure will be
given by the usual addition  ($x+\infty= +\infty$) and the constant $0$ 
which yield a commutative d-monoid. Every commutative monoid is
entropic. 

The d-monoid $(\oRp,+,0)$ has endomorphisms: For every $r\in\oRp$,
the map $x\mapsto rx\colon \oRp\to\oRp$ is a 
Scott-continuous endomorphism of the d-monoid $\oRp$ (for
$r=+\infty$ one agrees on $0\cdot(+\infty)= 0$ and
$r\cdot(+\infty) =+\infty$ for $r>0$ as usually in measure
theory); and every Scott continuous endomorphism of $\oRp$ is of this
form. The composition of two endomorphisms given by $r$ and $r'$ is the
endomorphisms given by $rr'$. Since $(\oRp,+,0)$ is entropic,
$\End(R)$ is also a commutative monoid with respect to
addition. Altogether, the algebra End$(R)$ is canonically isomorphic to 
semiring $(\oRp,+,\cdot, 0,1)$. The $\End(\oRp)$-d-modules
and module homomorphisms are precisely the d-cones and the linear maps
as introduced for example by  \cite{}. A cone is a commutative monoid
$C$ together with a 
scalar multiplication by nonnegative real numbers extended by
$+\infty$ satisfying the same axioms
as for vector spaces; in detail: 

\begin{definition}\label{def:cone}\index{cone}\index{ordered cone}{\rm
We take a signature
consisting of a constant $0$, unary operations $r\in \oRp$ and a
binary operation $+$. A \emph{cone} is an algebra of this
signature, that is, a set $C$ endowed with a
distinguished element $0$, an addition
$(x,y)\mapsto x+y\colon C\times C\to C$ and with a 
scalar multiplication 
$(r,x)\mapsto r\cdot x\colon \oRp\times C\to C$ satisfying for all
$x,y,z\in C$ and all $r,s\in\Rp$:}
\[\begin{array}{rcl}
x+(y+z)&=&(x+y)+z\\
x+y&=&y+x\\
x+0&=&x
\end{array}\]
  and 
\[\begin{array}{rcl} 
 1\cdot x&=&x\\
(rs)\cdot x&=&r\cdot (s\cdot x)\\
r\cdot (x+y)&=&r\cdot x+ r\cdot y\\
(r+s)\cdot x&=&r\cdot x+ s\cdot x\\
0\cdot x&=&0
\end{array}
\]
A map $f\colon C\to C'$ between cones is called \emph{linear}, if it
is additive and positively homogeneous, that is, if 
$$f(x+y)=f(x)+f(y) \ \ \mbox{ and } \ \ f(rx)=rf(x)$$
for all $x,y\in C$ and all $f\in\oRp$.  

If a cone $C$ is endowed with a directed complete partial order such
that addition and scalar multiplication are Scott-continuous, we have
a \emph{d-cone}. 
\end{definition} 

For a dcpo $X$, the function space $\oRp^X$ is a d-cone, too. We
denote by $\cV X$ the set of 
Scott-continuous linear maps $\mu\colon \oRp^X\to \oRp$. Note
that Scott-continuous additive maps between d-cones are easily shown
to be positively 
homogeneous, hence linear. Cones are entropic algebras. We infer that 
$\cV X$ is a d-cone, too, the order and 
the algebraic operations being defined pointwise. 

\begin{definition}\label{def:extprob}
The d-cone $\cV X$ is called the \emph{extended probabilistic
powerdomain} over $X$. 
\end{definition}

As a special case of \ref{prop:P-Q} we obtain: 

\begin{proposition}\label{prop:P-Q3}
There is a canonical one-to-one correspondence between state
transformers $t\colon X\to\cV Y$ and linear predicate transformers
$s\colon\oRp^Y\to\oRp^X$.
\end{proposition}

By Corollary \ref{cor:entropic}, $\cV X$ contains the d-subcone of
$[\oRp^X\to\oRp]$ generated 
by the projections $\delta_X(x), x\in X$, which are the classical
Dirac measures. But we do not have equality, in general, but we have
equality for an important subclass of dcpos. 

\begin{lemma}\cite[]{}\label{lem:approx}
If $X$ is a continuous dcpo, the d-cone $\cV X$ is continuous, too. In
fact, every $\varphi\in\cV X$ is the join of
a directed family of 'simple' valuations, that is, of finite linear
combinations of projections
$\sigma=\sum_{i=1}^nr_i\delta_X{x_i}$ with $\sigma\ll \varphi$. 
\end{lemma}
 
\begin{corollary}\label{cor:equal}
If $X$ is a continuous dcpo, the d-cone $\cV X$ of Scott-continuous
linear functionals $\varphi\colon \oRp^X\to\oRp$ equals the d-cone
generated by the projections.  
\end{corollary}

\section{Relaxed morphisms and relaxed entropic
  algebras}\label{sec:laxentropic} 

It is our aim to combine probability with nondeterminism. For this we
have to combine the semilattice structure for nondeterminism with the
additive structure for extended probability, that is, our algebra of
observations should be the extended reals $\oRp$ with two binary
operations $+$ and $\vee$ (or $\wedge$). As $+$ and $\vee$ do not
commute, we no longer have an entropic algebra. 
The framework developed in the previous sections is too
narrow. Surprisingly, one can deal with this situation by relaxing the
previous setting in replacing equalities by inequalities.

\begin{definition}\label{def:submorph}{\rm
Let $\omega$ be an operation of arity $n$ defined on dcpos $A$ and $A'$. A
Scott-continuous 
map $h\colon A\to A'$ is called an \emph{$\omega$-submorphism} if 
$$h\big(\omega(x_1,\dots,x_n)\big) \leq
\omega\big(h(x_1),\dots,h(x_n)\big) \ \mbox{ for all }
x_1,\dots,x_n\in A.$$ 
An \emph{$\omega$-supermorphism} is defined in the same way replacing
the inequality  $\leq$ by its opposite $\geq$. }
\end{definition}

For d-algebras of d-signature $\Omega$, we want to distinguish some 
operations $\omega\in\Omega$ for which we would like to consider
relaxed morphisms.  
For this, we suppose that each $\Omega_n$ is the union of two sub-dcpos
$\Omega_n^{\leq}$ and $\Omega_n^{\geq}$ which need not be disjoint. The
subsets  $\Omega_n^{\leq}$ and $\Omega_n^{\geq}$ will be 
kept fixed, and we let $\Omega^\leq =\bigcup_n\Omega_n^\leq$ and
$\Omega^\geq =\bigcup_n\Omega_n^\geq$. 

\begin{definition}\label{def:wmorph}{\rm
A Scott-continuous map $h\colon A\to A'$ between d-algebras of
d-signature $\Omega=\Omega^\leq\cup\Omega^\geq$ is said to be a 
\emph{relaxed d-morphism} if $h$ is an
$\omega$-submorphism for all $\omega\in\Omega^{\leq}$, but an
$\omega$-supermorphism for $\omega\in\Omega^{\geq}$. (For $\omega$ in
both $\Omega^{\leq}$ and $\Omega^{\geq}$, $h$ will be an
$\omega$-homomorphism.)}
\end{definition}

Of course, d-homomorphisms are also relaxed d-morphisms.
We record the following straightforward observation\display

\begin{lemma}\label{lem:comp}
The composition of relaxed d-morphisms between
d-algebras of d-signature $\Omega=\Omega^\leq\cup\Omega^\geq$ yields 
relaxed d-morphisms. 
\end{lemma}

For a d-algebra $R$ of d-signature $\Omega =\Omega^{\leq}\cup
\Omega^{\geq}$, we denote by $[R^X\lollipop_r R]$ the set of all relaxed  
d-morphisms $\varphi\colon R^X\to R$. The
pointwise supremum of a directed family of relaxed
d-morphisms is again a relaxed d-morphism. Thus,
these relaxed d-morphisms form a sub-dcpo of $[R^X\to R]$.  
As in the propositions \ref{prop:monad} and \ref{prop:P-Q} we have:

\begin{proposition}\label{prop:P-Q5}
Let $R$ be a d-algebra of d-signature $\Omega
=\Omega^{\leq}\cup\Omega^{\geq}$.  

{\rm (a)} For every state transformer $t\colon X\to [R^Y\to R]$
such that $t(x)$ is a relaxed d-morphism for each $x\in X$, the Kleisli
lifting $t^\dagger\colon [R^X\to R]\to [R^Y\to R]$ maps relaxed
d-morphisms to relaxed d-morphisms, so that our 
continuation monad $([R^-\to R],\delta,\dagger)$ restricts to a monad
$([R^X\lollipop_r R],\delta,\dagger)$.  

(b) Under the bijective correspondences $P$ and $Q$ (see lemma
\ref{lem:PQ}, the predicate 
 transformers corresponding to the state transformers $t\colon X\to
 [R^Y\lollipop_r R]$ are the relaxed d-morphisms $s\colon R^Y\to R^X$,
 that is: 
$$  [R^Y\to R]^X \ \cong \ [R^Y\lollipop_r R^X]$$
\end{proposition}

The proofs are the same as for the corresponding claims in
\ref{lem:kleisli}(c) and \ref{prop:P-Q}. We just 
have to replace the equality by the appropriate inequality ($\leq$ in case
$\omega\in \Omega^{\leq}$ and $\geq$ in case $\omega\in
\Omega^{\geq}$) every time that we have used the homomorphism property there. 

We now turn to the question under what circumstances, the
relaxed d-morphisms form a subalgebra of $[R^X\to R]$. 

We attack this question more generally and consider a d-algebra $B$ of
d-signature $\Omega=\Omega^{\leq}\cup\Omega^{\geq}$ and we ask the
question, whether the set of 
relaxed d-morphisms $\varphi\colon B\to R$ is a subalgebra
of $[B\to R]$. 

In order to answer this question we consider an operation $\sigma\in
\Omega$ of arity $n$ and we have to show that
$\sigma(\varphi_1,\dots,\varphi_n)$ is a relaxed morphism for all
relaxed d-morphisms $\varphi_1,\dots, 
\varphi_n\colon B\to R$, that is, for all $\omega\in \Omega^\leq$ of
arity $m$ and all $f_1,\dots,f_m\in B$, we have: 
\begin{eqnarray}\label{subhomo}
\sigma(\varphi_1,\dots,\varphi_n)(\omega(f_1,\dots,f_m)) \leq  
\omega(\sigma(\varphi_1,\dots,\varphi_n)(f_1),\ \dots\ ,
\sigma(\varphi_1,\dots,\varphi_n)(f_m)),
\end{eqnarray}
and analogously, with the reverse inequality, for $\omega\in\Omega^{\geq}$.

\begin{definition}\label{def:subcommutes}{\rm 
We will say that an operation $\sigma$ of arity $n$ on a d-algebra $A$
\emph{subcommutes} with an operation $\omega$ of arity $m$
(equivalently, $\omega$ \emph{supercommutes} 
with $\sigma$) if, for all $x_{ij}\in A, i=1,\dots n,\
j=1,\dots,m$\display 
\begin{eqnarray}\label{subdistributes}
\sigma(\omega(x_{11},\dots,x_{1m}),\ \dots\ ,\omega(x_{n1},\dots,x_{nm})) \leq
\omega(\sigma(x_{11},\dots,x_{n1}),\ \dots\ ,
\sigma(x_{1m},\dots,x_{nm})).
\end{eqnarray}  }
\end{definition}

This is equivalent to the statement that $\sigma$ is an
$\omega$-submorphism.
Whenever this subcommutatitivity law holds in $R$, it also holds in
$R^X$ and $R^{R^X}$ and in subalgebras thereof. As a consequence,  
inequation (\ref{subhomo}) holds if $\sigma$ subcommutes with
$\omega$.  Indeed,
\[\begin{array}{rcl}
\sigma(\varphi_1,\dots,\varphi_n)(\omega(f_1,\dots,f_m)) &= &
\sigma\big(\varphi_1(\omega(f_1,\dots,f_m)),\ \dots\
,\varphi_n(\omega(f_1,\dots,f_m))\big)\\ &&\mbox{(the operation
  $\sigma$ being defined pointwise)}\\ 
 &\leq& \sigma\big(\omega(\varphi_1(f_1),\dots,\varphi_1(f_m)),\ 
\dots\ ,\omega(\varphi_n(f(1),\dots,\varphi_n(f_m))\big)\\&&\mbox{\rm (the
  $\varphi_i$ being  $\omega$-submorphisms)}\\ 
&\leq&\omega\big(\sigma(\varphi_1(f_1),\dots,\varphi_n(f_1)),\ \dots\
,\sigma(\varphi_1(f_m),\dots,\varphi_n(f_m))\big)\\ 
&&\mbox{(since $\sigma$ subcommutes with $\omega$, equation
  (\ref{subdistributes})} 
\end{array}\]
This leads to the following definition\display

\begin{definition}\label{def:weaklyentropic}{\rm
We will say that the d-algebra $R$ of d-signature
$\Omega=\Omega^\leq\cup\Omega^\geq$ is \emph{relaxed entropic}, if
every $\sigma\in\Omega$ is a relaxed morphism, that is, if $\sigma$
subcommutes with every $\omega\in 
\Omega^{\leq}$ and supercommutes with every
$\omega\in\Omega^{\geq}$. }
\end{definition}

In a relaxed entropic algebra, any two $\sigma,\omega\in\Omega_\leq$
commute and similarly for $\sigma,\omega\in\Omega_\geq$.
By the arguments preceding the definition we have:

\begin{proposition}\label{prop:subalg}
If $R$ is a relaxed entropic d-algebra, the 
relaxed d-morphisms from any d-algebra $B$ to $R$  form a subalgebra
of $[B\to R]$. 
\end{proposition}

Specializing to $B=R^X$ for a dcpo $X$ gives us\display

\begin{corollary}\label{cor:subhomo}
If $R$ is a relaxed entropic d-algebra,  the collection
$[R^X\lollipop_r R]$ of all relaxed d-morphisms $\varphi\colon R^X\to
R$ is a d-subalgebra of $[R^X\to R]$.  
\end{corollary}

We now turn to the question whether $[R^X\lollipop_r R]$ equals the 
d-algebra $\cF_R X$ generated by the projections. Again we do not have a
general answer, but at least a containment relation:
 
\begin{proposition}\label{prop:free6}
If $R$ is a relaxed entropic d-algebra, then the d-subalgebra $\cF_R X$
of $[R^X\to R]$ generated by the projections $\delta_X(x), x\in X,$ is a
d-subalgebra of $[R^X\lollipop_r R]$.
\end{proposition}

\begin{proof}
The projections $\delta_X(x)$ are d-homomorphisms, hence
also relaxed d-morphisms and, in particular, contained in
$[R^X\lollipop_r R]$.  If $R$ is
relaxed entropic, then $[R^X\lollipop_r R]$ is a d-subalgebra of
$[R^X\to R]$ by Proposition \ref{subhomo}. Thus, $[R^X\lollipop_r R]$ 
contains the d-subalgebra $F_R X$ which is generated by the projections. 
\end{proof}

Whether we have equality $F_R X =[R^X\lollipop_r R]$, has to be
decided in each special case separately.

\section{Example: Combining nondeterminism and extended probability}  
\label{sec:combining}

We now combine extended probability and nondeterminism that had been
considered separately in Section \ref{sec:nondet}. For this purpose
our domain $R$ of observations will be 
$\oRp=\{r\in\mathbb R\mid r\geq 0\}\cup\{+\infty\}$ endowed both with
its cone structure and a semilattice operation, $\max$ or 
$\min$, which we denote by $\vee$ and $\wedge$, respectively. 

\subsection{The angelic case}\label{subsec:angelic1}

We first look at $\oRp$ as a d-cone (see Definition \ref{def:cone})
with $\vee$ as an additional binary operation. Considered separately
as a $\vee$-semilattice and as a cone, 
$\oRp$ is entropic. Also multiplication by scalars commutes with
the binary operation $\vee$; indeed, 
$$r(x\vee y) = rx\vee ry$$
holds for all $r\in\oRp$ and all $x,y\in\oRp$. But the two
binary operations $+$ and $\vee$ do not not commute so that
$(\oRp,+,\vee,r\cdot-,0)$ is not entropic. But the inequational law \display 
\begin{eqnarray}
(x_1+y)\vee(x_2+z)\leq (x_1\vee x_2)+(y\vee z), \label{Sub}
\end{eqnarray}
holds for all $x_1,x_2,y,z\in\oRp$, that is, $\vee$ subcommutes
with $+$ (and $+$ supercommutes with 
$\vee$). Thus, we put $\vee$ into $\Omega^\leq$ and $+$ in $\Omega^\geq$;
the constant $0$ and the unary operations $x\to rx$ will be both in
$\Omega^\leq$ and $\Omega^\geq$, and we have:

\begin{lemma}\label{lem:relaxed}
$\oRp$ with the binary operations $+,\vee$, the unary operations
$x\mapsto rx, r\in\oRp,$ and the constant $0$ is a relaxed entropic
d-algebra. 
\end{lemma}

According to the general procedure in the previous section, let $X$ be
any dcpo and consider the relaxed 
d-morphisms $\varphi\colon \oRp^X\to\oRp$. In view of our signature
$\Omega = \Omega^\leq\cup \Omega^\geq$, such a relaxed d-morphism is
characterized by the following equalities and inequalities: 
\begin{eqnarray}
\varphi(0)&=&0,\\
\varphi(rf)&=&r\varphi(f),\\ \varphi(f+g)&\leq&\varphi(f)+\varphi(g)\\
\varphi(f\vee g)&\geq&\varphi(f)\vee\varphi(g)
\end{eqnarray} 
for all $f,g\in\oRp^X$ and all
$r\in\oRp$.  Since the first equality is a consequence of the second
for $r=0$, and since the last inequality is
always satisfied for order preserving maps, we can omit these two and we see
that the relaxed morphisms are nothing but the \emph{sublinear
functionals}. We denote by $\cH\cV X$ the set of all these
Scott-continuous sublinear functionals $\varphi$.

Since linear functionals are sublinear, $\cH\cV X$ contains $\cV X$.
Moreover, $\cH X$ is a retract of $\cH\cV X$: just use the retraction of
$\oRp$ onto $\{0,+\infty\}$ mapping all $r>0$ onto $+\infty$ and notice
that $+$ and $\vee$ agree on  $\{0,+\infty\}$.

As a special case of  Corollary \ref{cor:subhomo} and Proposition
\ref{prop:P-Q5} we obtain:  

\begin{proposition}\label{prop:mixedang}
(a) For every dcpo $X$, the set $\cH\cV X$ of Scott-continuous sublinear
functionals $\varphi\colon \oRp^X\to\oRp$ is a d-subcone and a
d-$\vee$-subsemilattice of
$[\oRp^X\to\oRp]$; in fact, $(\cH\cV X,\delta,{}^\dagger)$ is a monad over
the category {\sf DCPO} subordinate to the continuation monad.

(b) There is a canonical one-to-one correspondence between state
transformers $t\colon X\to\cH\cV Y$ and sublinear predicate transformers
$s\colon\oRp^Y\to\oRp^X$.
\end{proposition}

By Proposition \ref{prop:mixedang}(a), $\cH\cV X$ is a subalgebra of
$[\oRp^X\to \oRp]$ and by Proposition \ref{prop:free6} it contains the d-cone
$\vee$-subsemilattice $\cF_{\oRp} X$ generated by the projections $\wh
x, x\in X$. For dcpos $X$ in general, equality need not hold.
For this we have to require continuity::

\begin{proposition} 
For every continuous dcpo $X$, the subalgebra $\cF_{\oRp} X$ of  $[\oRp^X\to
\oRp]$ generated by the projections agrees with the d-cone
$\vee$-semilattice $\cH\cV X$ of all Scott-continuous sublinear functionals 
$\varphi\colon \oRp^X\to\oRp$.
\end{proposition}

\begin{proof} (Sketch) This proposition has been proved in
  \cite{KP,TKP}. The proof uses the 
following ingredients: Let $X$ be a continuous dcpo. Then $\oRp^X$ is
a continuous d-cone. We firstly use a
Hahn-Banach type theorem \cite[5.9(1)]{KP} that tells us that every
Scott-continuous sublinear functional $\varphi\colon \oRp^X\to\oRp$ is
pointwise the sup of a family of Scott-continuous linear functionals
$\mu_i\colon  \oRp^X\to\oRp$. In Lemma \ref{lem:approx} we have seen
that every $\mu_i$ is pointwise the supremum of a directed family of
finite linear combinations $\sigma_{ij}=\sum_{k=1}^nr_{kj}\delta_X(x_{kj})$ of
projections (Dirac measures). Thus, $\varphi$ is the supremum of the
$\sigma_{ij}$. Taking finite suprema of the $\sigma_{ij}$ one obtains
a directed family $\bigvee_{j=1}^m\sum_{k=1}^nr_{kj}\delta_X(x_{kj})$ of
finite suprema of simple valuations. This shows that $\varphi$ belongs
to the d-cone $\vee$-subsemilattice generated by the projections.
\end{proof}

In \cite{KP,TKP} one finds another representation of  
$\cH\cV X$, namely as the collection of all nonempty Scott-closed
convex subsets of the d-cone $\cV X$ ordered by inclusion. The equivalence
of the two presentation is given as follows: To every sublinear
functional $\varphi\colon \oRp^X\to\oRp$ we assign the set
$C(\varphi)$ of all linear functionals $\mu\in\cV X$ such that
$\mu\leq\varphi$. Then 
$C(\varphi)$ is a nonempty closed convex subset of $\cV X$. In
\cite[Corollary 6.3]{KP} it is shown that the assignment
$\varphi\mapsto C(\varphi)$ is an order isomorphism from $\cH\cV X$
onto the collection of nonempty Scott-closed convex subsets of $\cV X$
ordered by inclusion.

\subsection{The demonic case}\label{subsec:demonic1}

We now look at $\oRp$ as a d-cone with the additional binary operation
$\wedge$ ($x\wedge y=\min(x,y)$). 
As in the angelic case this algebra is not entropic, since addition and 
meet do not commute. But we have the inequation
\begin{eqnarray}
(x_1+y)\wedge(x_2+z)\geq (x_1\wedge x_2)+(y\wedge z),
\end{eqnarray}
that is, $\wedge$ supercommutes with $+$ (and $+$ subcommutes with
$\wedge$). Thus, we put $+$ into $\Omega^\leq$ and $\wedge$ in $\Omega_\geq$;
the constant $0$ and unary operations $x\to rx$ will be both in
$\Omega^\leq$ and $\Omega^\geq$, and we have:

\begin{lemma}\label{lem:relaxed1}
$\oRp$ with the binary operations $+,\wedge$, the unary operations
$x\mapsto rx, r\in\oRp,$ and the constant $0$ is a relaxed entropic
d-algebra. 
\end{lemma} 
 
Similarly as in the angelic case, the relaxed morphisms $\varphi\colon
\oRp^X\to\oRp$ are characterizeb by
\begin{eqnarray}
\varphi(rf)&=&r\varphi(f),\\ 
\varphi(f+g)&\geq&\varphi(f)+\varphi(g),\\ 
\end{eqnarray} 
for all $f,g\in\oRp^X$ and all $r\in\oRp$.  We see
that the relaxed d-morphisms are nothing but the Scott-continuous superlinear
functionals. We denote by $\cS\cV X$ the set of all these superlinear
functionals $\varphi$.
Since linear functionals are superlinear, $\cS\cV X$ contains $\cV(X)$.

As a special case of  Corollary \ref{cor:subhomo} and Proposition
\ref{prop:P-Q5} we obtain:  

\begin{proposition}\label{prop:mixeddem}
(a) For every dcpo $X$, the set $\cS\cV X$ of Scott-continuous superlinear
functionals $\varphi\colon \oRp^X\to\oRp$ is a d-subcone and a
d-$\wedge$-subsemilattice of
$[\oRp^X\to\oRp]$; in fact, $(\cS\cV X,\delta,{}^\dagger)$ is a monad over
the category {\sf DCPO} subordinate to the continuation monad.

(b) There is a canonical one-to-one correspondence between state
transformers $t\colon X\to\cS\cV Y$ and superlinear predicate transformers
$s\colon\oRp^Y\to\oRp^X$.
\end{proposition}
 
By Proposition \ref{prop:mixeddem}(a), $\cS\cV X$ is a d-cone
$\wedge$-subsemilattice of
$[\oRp^X\to \oRp]$ and by Proposition \ref{prop:free6} it contains the d-cone
$\wedge$-subsemilattice $\cF_{\oRp} X$ generated by the projections $\wh
x, x\in X$. For dcpos $X$ in general, equality need not hold.
For this we have to require even more than continuity:

\begin{proposition} 
For every continuous coherent dcpo $X$, the d-cone $\wedge$-subsemilattice
$\cF_{\oRp} X$ of  $[\oRp^X\to \oRp]$ generated by the projections
agrees with the d-cone 
$\wedge$-semilattice $\cS\cV X$ of all Scott-continuous superlinear
functionals  $\varphi\colon \oRp^X\to\oRp$.
\end{proposition}

The coherence property required in the previous proposition means that
the intersection of any two Scott-compact saturated subsets of $X$ is
Scott-compact. 
This proposition follows from results in
  \cite{KP,TKP}. One uses that, for a continuous coherent dcpo $X$, the
  function space $\oRp^X$ is a coherent continuous d-cone. A 
Hahn-Banach type theorem \cite[5.9(2)]{KP} tells us that every
Scott-continuous superlinear functional $\varphi\colon \oRp^X\to\oRp$ is
pointwise the infimum of a family of Scott-continuous linear functionals
$\mu\colon  \oRp^X\to\oRp$. But now the arguments become more
sophisticated than in the angelic case. One has to show that $\varphi$
is pointwise the supremum of a directed family of finite infima of
finite linear combinations of projections.

In \cite{KP,TKP} one finds another representation of  
$\cS\cV X$, namely as the collection of all nonempty Scott-compact saturated
convex subsets of the d-cone $\cV X$ ordered by inclusion. The equivalence
of the two presentation is given as follows: To every Scott-continuous
superlinear functional $\varphi\colon \oRp^X\to\oRp$ we assign the set
$K(\varphi)$ of all linear functionals $\mu\in\cV X$ such that
$\mu\geq\varphi$. Then 
$K(\varphi)$ is a nonempty Scott-compact saturated convex subset of $\cV X$. In
\cite[Corollary 6.6]{KP} it is shown that the assignment
$\varphi\mapsto K(\varphi)$ is an order isomorphism from $\cS\cV X$
onto the collection of nonempty Scott-compact saturated convex subsets
of $\cV X$ ordered by reverse inclusion.

\end{document}